\documentclass[pra,twocolumn,superscriptaddress]{revtex4}
\usepackage{amsmath,color,graphics,graphicx,theorem}
\usepackage{pifont}
\usepackage{amsfonts}
\usepackage{multirow}
\usepackage{rotating}
\usepackage{array}
\usepackage{enumerate}

\usepackage[utf8]{inputenc}
\usepackage[T1]{fontenc}
\usepackage{lmodern}

\newtheorem{definition}{Definition}
\newtheorem{proposition}[definition]{Proposition}
\newtheorem{lemma}[definition]{Lemma}

\newtheorem{theorem}[definition]{Theorem}
\newtheorem{corollary}[definition]{Corollary}
\newtheorem{conjecture}[definition]{Conjecture}

\newtheorem{remark}[definition]{Remark}
\newtheorem{example}[definition]{Example}
\newtheorem{question}[definition]{Question}

\def\squareforqed{\hbox{\rlap{$\sqcap$}$\sqcup$}}
\def\qed{\ifmmode\squareforqed\else{\unskip\nobreak\hfil
\penalty50\hskip1em\null\nobreak\hfil\squareforqed
\parfillskip=0pt\finalhyphendemerits=0\endgraf}\fi}
\def\endenv{\ifmmode\;\else{\unskip\nobreak\hfil
\penalty50\hskip1em\null\nobreak\hfil\;
\parfillskip=0pt\finalhyphendemerits=0\endgraf}\fi}
\newenvironment{proof}{\noindent \textbf{{Proof.~} }}{\qed}
\def\Dbar{\leavevmode\lower.6ex\hbox to 0pt
{\hskip-.23ex\accent"16\hss}D}
\makeatletter
\def\url@leostyle{%
  \@ifundefined{selectfont}{\def\UrlFont{\sf}}{\def\UrlFont{\small\ttfamily}}}
\makeatother
\def\bcj{\begin{conjecture}}
\def\ecj{\end{conjecture}}
\def\bcr{\begin{corollary}}
\def\ecr{\end{corollary}}
\def\bd{\begin{definition}}
\def\ed{\end{definition}}
\def\bem{\begin{enumerate}}
\def\eem{\end{enumerate}}
\def\bex{\begin{example}}
\def\eex{\end{example}}
\def\bim{\begin{itemize}}
\def\eim{\end{itemize}}
\def\bl{\begin{lemma}}
\def\el{\end{lemma}}
\def\bpf{\begin{proof}}
\def\epf{\end{proof}}
\def\bpp{\begin{proposition}}
\def\epp{\end{proposition}}
\def\bqu{\begin{question}}
\def\equ{\end{question}}
\def\br{\begin{remark}}
\def\er{\end{remark}}
\def\bt{\begin{theorem}}
\def\et{\end{theorem}}
\def\btb{\begin{tabular}}
\def\etb{\end{tabular}}

\newcommand{\nc}{\newcommand}

\def\a{\alpha}\def\b{\beta}\def\g{\gamma}\def\r{\rho}\def\s{\sigma}\def\ph{\varphi}\def\c{\chi}\def\ps{\psi}\def\o{\omega}

\def\G{\Gamma}\def\L{\Lambda}

\nc{\bA}{{\bf A}} \nc{\bB}{{\bf B}} \nc{\bC}{{\bf C}} \nc{\bD}{{\bf D}} \nc{\bE}{{\bf E}} \nc{\bF}{{\bf F}} \nc{\bG}{{\bf G}} \nc{\bH}{{\bf H}} \nc{\bI}{{\bf I}} \nc{\bJ}{{\bf J}} \nc{\bK}{{\bf K}} \nc{\bL}{{\bf L}} \nc{\bM}{{\bf M}} \nc{\bN}{{\bf N}} \nc{\bO}{{\bf O}} \nc{\bP}{{\bf P}} \nc{\bQ}{{\bf Q}} \nc{\bR}{{\bf R}} \nc{\bS}{{\bf S}} \nc{\bT}{{\bf T}} \nc{\bU}{{\bf U}} \nc{\bV}{{\bf V}} \nc{\bW}{{\bf W}} \nc{\bX}{{\bf X}}
\nc{\bZ}{{\bf Z}}

\nc{\cA}{{\cal A}} \nc{\cB}{{\cal B}} \nc{\cC}{{\cal C}}
\nc{\cD}{{\cal D}} \nc{\cE}{{\cal E}} \nc{\cF}{{\cal F}}
\nc{\cG}{{\cal G}} \nc{\cH}{{\cal H}} \nc{\cI}{{\cal I}}
\nc{\cJ}{{\cal J}} \nc{\cK}{{\cal K}} \nc{\cL}{{\cal L}}
\nc{\cM}{{\cal M}} \nc{\cN}{{\cal N}} \nc{\cO}{{\cal O}}
\nc{\cP}{{\cal P}} \nc{\cQ}{{\cal Q}} \nc{\cR}{{\cal R}}
\nc{\cS}{{\cal S}} \nc{\cT}{{\cal T}} \nc{\cU}{{\cal U}}
\nc{\cV}{{\cal V}} \nc{\cW}{{\cal W}} \nc{\cX}{{\cal X}}
\nc{\cZ}{{\cal Z}}

\nc{\hA}{{\hat{A}}} \nc{\hB}{{\hat{B}}} \nc{\hC}{{\hat{C}}}
\nc{\hD}{{\hat{D}}} \nc{\hE}{{\hat{E}}} \nc{\hF}{{\hat{F}}}
\nc{\hG}{{\hat{G}}} \nc{\hH}{{\hat{H}}} \nc{\hI}{{\hat{I}}}
\nc{\hJ}{{\hat{J}}} \nc{\hK}{{\hat{K}}} \nc{\hL}{{\hat{L}}}
\nc{\hM}{{\hat{M}}} \nc{\hN}{{\hat{N}}} \nc{\hO}{{\hat{O}}}
\nc{\hP}{{\hat{P}}} \nc{\hR}{{\hat{R}}} \nc{\hS}{{\hat{S}}}
\nc{\hT}{{\hat{T}}} \nc{\hU}{{\hat{U}}} \nc{\hV}{{\hat{V}}}
\nc{\hW}{{\hat{W}}} \nc{\hX}{{\hat{X}}} \nc{\hZ}{{\hat{Z}}}

\def\bisep{\textup{BS}}

\def\FS{\textup{FS}}

\def\dim{\mathop{\rm Dim}}

\def\min{\mathop{\rm min}}

\def\tr{\mathop{\rm Tr}}

\def\GL{{\mbox{\rm GL}}}

\def\bigox{\bigotimes}
\def\dg{\dagger}

\def\ox{\otimes}

\def\su{\subset}
\def\sue{\subseteq}

\newcommand{\bra}[1]{\langle#1|}
\newcommand{\ket}[1]{|#1\rangle}
\newcommand{\proj}[1]{| #1\rangle\!\langle #1 |}
\newcommand{\ketbra}[2]{|#1\rangle\!\langle#2|}

\newcommand{\norm}[1]{\lVert#1\rVert}
\newcommand{\abs}[1]{|#1|}
\newcommand{\ot}{\otimes}

\newcommand{\jmp}{J. Math. Phys.~}

\begin{document}
\title{Role of correlations in the two-body-marginal problem}

\author{Lin Chen}
\email{linchen0529@gmail.com}
\affiliation{Singapore University of Technology and Design, 20 Dover Drive, Singapore
138682}
\affiliation{Institute for Quantum Computing \& Department of Pure Mathematics,
University of Waterloo, 200 University Avenue West, N2L 3G1 Waterloo, Ontario, Canada}

\author{Oleg Gittsovich}
\email{oleg.gittsovich@univie.ac.at}
\affiliation{Institute for Quantum Computing \& Department of Physics and Astronomy,
University of Waterloo, 200 University Avenue West, N2L 3G1 Waterloo, Ontario, Canada}
\affiliation{Institute for Quantum Optics and Quantum Information, Austrian Academy of Sciences, Technikerstr. 21a, A-6020 Innsbruck, Austria}
\affiliation{Institute for Theoretical Physics, University of Innsbruck, Technikerstr. 25, A-6020 Innsbruck, Austria}

\author{K. Modi}
\email{kavan.modi@monash.edu}
\affiliation{School of Physics, Monash University, Victoria 3800, Australia}

\author{Marco Piani}
\email{mpiani@uwaterloo.ca}
\affiliation{Institute for Quantum Computing \& Department of Physics and Astronomy,
University of Waterloo, 200 University Avenue West, N2L 3G1 Waterloo, Ontario, Canada}
\affiliation{Department of Physics and SUPA, University of Strathclyde, Glasgow G4 0NG, UK}

\begin{abstract}
Quantum properties of correlations
have a key role in disparate fields of physics, from quantum
information processing, to quantum foundations, to strongly
correlated systems. We tackle a specific aspect of the fundamental
quantum marginal problem: we address the issue of deducing the
global properties of correlations of tripartite quantum states based
on the knowledge of their bipartite reductions, focusing on relating
specific properties of bipartite correlations to global correlation
properties. We prove that strictly classical bipartite correlations
may still require global entanglement and that unentangled---albeit
not strictly classical---reductions may require global genuine
multipartite entanglement, rather than simple entanglement. On the
other hand, for three qubits, the strict classicality of the
bipartite reductions rules out the need for genuine multipartite
entanglement. Our work sheds new light on the relation between local
and global properties of quantum states, and on the interplay
between classical and quantum properties of correlations.
\end{abstract}

\date{ \today }

\maketitle

\section{Introduction}

Quantum correlations have a central role in
quantum information processing, in quantum foundations, as well as
in the physics of strongly correlated systems~\cite{hhh09, mbc12,
bcp13,revarealaw}. On one hand, quantum correlations, and in
particular entanglement, are a resource that allows one to go beyond
what classically possible in many scenarios, from communication
tasks, to (measurement-based) quantum computing, to quantum
cryptography. On the other hand, the non-classicality of quantum
correlations---be it in the form of non-locality, steering,
entanglement, or discord---is one of the most distinctive traits of
quantum mechanics, and challenges our understanding of quantum
mechanics itself. The interplay between local and global properties
of quantum states is a key aspect in the study of quantum
correlations, both from a fundamental perspective and an applicative
one. For example, we may want to certify the presence of
multipartite entanglement in large systems without the---often
inaccessible---knowledge of the global state, using instead only the
information that comes from reduced states. On the other hand, in
condensed-matter physics, because of the typically local---e.g.,
two-body---interactions, relevant properties are dictated by the
interplay between the allowed reduced states and global
correlations, giving raise to phenomena like
frustration~\cite{frustration}. The general study of the relations
between the properties of the reduced states and the properties of
the global state is known as the quantum marginal problem, which has
seen a growing interest in the past years also for the reasons
above~\cite{bravyi04a, klyachko04, christandl06,
christandl12, chen12, jv2013, cjk13}.

In this work we study what can be inferred about the quality of
the correlations of the global state given information about the
two-body reduced states, aiming at answering the question: \emph{What correlations need to be present globally to explain what we see locally?} In~\cite{christandlscience} a
characterization of multipartite entanglement in terms of even just
single-party reduced states (actually, single-party spectra) was
given, but under the assumption of dealing with a pure or quasi-pure
global state. In~\cite{guhne1, guhne2, wba12} the possibility of
dealing with global mixed states is taken into account, and examples
are given where two-qubit separable states are only compatible with
global entanglement, intended in the sense of lack of total
separability (see Section~\ref{sec:definitions} for definitions). In~\cite{wba12} examples
are also given where genuine multipartite entanglement---a much
stronger notion of global entanglement---can still be deduced from
the properties of the two-body reduced states, but only when these
reduced states exhibit bipartite entanglement themselves.

In this work we present several results that complement and generalize those of~~\cite{guhne1, guhne2, wba12}.
We offer a brief summary of our findings in Table~\ref{table}.
Firstly, we provide examples of triples of bipartite reduced states---in
the simplest case, two-qubit states---that, albeit separable, are
only compatible with genuine tripartite entanglement (lower-right corner of Table~\ref{table}). As far as we
know, this is the strongest ``gap'' known between the entanglement
properties of the marginals and of the global state. Secondly, we
address the issue of relating the general quantumness of
correlations~\cite{mbc12} of reduced states to the quantum
correlations of the global state. We find that strictly classical
reduced states may still be compatible only with global
entanglement (upper-left corner of Table~\ref{table}). Thirdly, we find that, at least for qubits,
the strict classicality of the two-body correlations makes it
impossible to certify genuine tripartite entanglement based on the
knowledge of the reductions: strictly classical two-qubit reduced
states are always compatible with a global state that is not genuine
tripartite entangled (upper-right corner of Table~\ref{table}).

The rest of the paper is organized as follows. In Section~\ref{sec:definitions}  we define the relevant notions of correlations and classicality, and of compatibility of two-body reduced states in tripartite systems. In Section~\ref{sec:classicalcomp} we study the relation between the classicality of reductions and their compatibility. In Section~\ref{sec:genuine} we prove that unentangled reduced states may only be compatible with genuine multipartite entanglement at the level of the global state. 
Finally, we conclude in Section~\ref{sec:conclusions}.

\begin{table}
\begin{tabular}{l l|>{\centering}m{1.5cm}|>{\centering}m{3.5cm}|}
\cline{3-4}
& & \multicolumn{2}{c|}{Global state}\\
\cline{3-4}
 &  & {Entangled  } & \begin{minipage}{3.5cm}{\vspace{2mm} Genuinely multipartite\\ entangled}\vspace{2mm}\end{minipage} \cr
  \hline
\multicolumn{1}{|c|}
{\multirow{2}{*}{\begin{sideways}{Reductions}\end{sideways}}}
&\begin{minipage}{1.5cm}\vspace{2mm}Fully\\ classical \vspace{2mm}\end{minipage} & \ding{52} & \ding{56} (qubits)\cr
\cline{2-4}
\multicolumn{1}{|c|}{}
&\begin{minipage}{1.5cm}\vspace{4mm}Separable\vspace{4mm} \end{minipage}   & \ding{51}~\cite{guhne1, guhne2, wba12} & \ding{52}  \cr
\cline{2-4}
\hline
\end{tabular}
\caption{A \ding{51} in a cell means that there exist two-body marginal states with the quality of correlations listed in the corresponding row, that are only compatible with global states that have \emph{at least} the property of correlations listed in the corresponding column. A \ding{55} means that the inference is not possible: specifically, there are no fully classical two-qubits states that are only compatible with genuine tripartite entanglement (refer to the main text for definitions). Our results correspond to bold symbols, \ding{52} and \ding{56}. Previous results are reported for completeness and comparison. \label{table}}
\end{table}

\section{Correlations and compatibility}
\label{sec:definitions}

We begin by formally defining qualitatively different types of correlations.
\bd\label{def:separable} Any tripartite mixed state can be written
as a mixture of an ensemble of pure states as $\r_{ABC}=\sum_{i} p_i
\ket{\psi_i} \bra{\psi_i}_{ABC}$. We say that $\r_{ABC}$ is:
\begin{itemize}
\item
\emph{fully separable}, if we can take each $\ket{\psi_i}_{ABC}$ to
be fully factorized, e.g. $\ket{\alpha_i}_A
\ket{\beta_i}_B\ket{\gamma_i}_C$;
\item
\emph{biseparable}, if we can take each $\ket{\psi_i}_{ABC}$ to be
unentangled in at least one partition, e.g.,
$\ket{\alpha_i}_A\ket{\phi_i}_{BC}$, $\ket{\beta_i}_B
\ket{\phi}_{AC}$, or $\ket{\gamma_i}_C\ket{\phi}_{AB}$;
\item
\emph{genuinely multipartite entangled}, if for any ensemble there
is at least one $\ket{\psi_i}$ with $p_i>0$ that is not factorized
with respect to any bipartition, i.e., if $\r_{ABC}$ is not
biseparable;
\item
\emph{fully classical}, if we can take each $\ket{\psi_i}_{ABC}$ to
be of the form $\ket{a_i}_A\ket{b_j}_B\ket{c_k}_C$, with
$\{\ket{a_i}\},\; \{\ket{b_j}\} \; \{\ket{c_k}\}$ orthonormal basis
on $\cH_A$,  $\cH_B$, and $\cH_C$, respectively, so that, overall,
$\r_{ABC}=\sum_{ijk} p_{ijk} \proj{a_i b_j c_k}$.
\end{itemize}
Bipartite full classicality and separability are defined similarly:
$\rho_{AB}$ is \emph{fully classical} if $\rho_{AB}=\sum_{ij} p_{ij}
\proj{a_i b_j}$, for $\{\ket{a_i}\},\; \{\ket{b_j}\}$ orthonormal
bases, and \emph{separable} if $\rho_{AB}= \sum_i p_i \proj{\alpha_i
\beta_i}$. A  bipartite state is \emph{entangled} if it is not separable.
The notions of full separability and biseparability are redundant
for bipartite states. \ed

The set of fully classical states, $\mathcal{S}_{\textrm{FC}}$, is a
subset of the set of fully separable states,
$\mathcal{S}_{\textrm{FS}}$, which in turn is a subset of the set of
biseparable states, $\mathcal{S}_{\textrm{BS}}$; the set of
genuinely multipartite entangled states,
$\mathcal{S}_{\textrm{GME}}$ is the is the complement of
$\mathcal{S}_{\textrm{BS}}$ in the space of all states,
$\mathcal{S}_{\textrm{ALL}}$ (see Figure~\ref{fig:sets}). All the
mentioned sets apart from $\mathcal{S}_{\textrm{FC}}$ and
$\mathcal{S}_{\textrm{GME}}$ are  convex.  A biseparable state may
be either separable or entangled in any bipartition,  but it is by
definition the convex sum of three states that are each separable in
one of the bipartitions $A:BC$, $B:AC$, and $C:AB$.
\begin{figure}
\includegraphics[width=0.4 \textwidth]{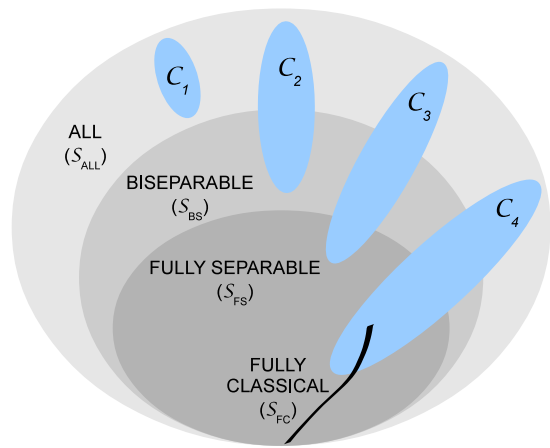}
\caption{Hierarchy of correlation classes, and some possible
compatibility sets.  The set $\mathcal{S}_{\textrm{FC}}$ is denoted
by the line. The set $\mathcal{S}_{\textrm{GME}}$ of genuinely multipartite states is the complement of
$\mathcal{S}_{\textrm{BS}}$ in the space of all states, $\mathcal{S}_{\textrm{ALL}}$. The reductions corresponding to the compatibility set
$\mathcal{C}_1$ are only compatible with genuine multipartite
entanglement. The reductions of the states in compatibility set
$\mathcal{C}_2$ are compatible with entangled biseparable
states and genuinely multipartite entangled states, but not separable states. The reductions defining $\mathcal{C}_3$ are compatible  with
a fully separable state as well as with entangled states, but not
fully-classical states.} \label{fig:sets}
\end{figure}
We now move to define formally the notion of compatibility for reduced states.

\bd \label{def:compatible}  Given a triple of three two-party states
$\mathcal{E}=(\rho_{AB},\rho_{AC},\rho_{BC})$, its
\emph{compatibility set} is defined as
$\mathcal{C}(\mathcal{E}):=\{\sigma_{ABC}\in\mathcal{S}_\textrm{ALL}|\sigma_{ij}=\rho_{ij},
ij=AB,AC,BC\}$.  Any compatibility set is a convex
set~\cite{chen12}, and the property of being part of a given
compatibility set defines an equivalence relation. We find it useful
to denote by $\mathcal{C}(\rho_{ABC})$ the compatibility set
associated with the reduced states of $\rho_{ABC}$, i.e., the set of
all states that have the same reductions as $\rho_{ABC}$.
A triple of two-party states
$\mathcal{E}$ is said to be \emph{compatible} (so that we refer to
the triple as \emph{triple of reductions}) if
$\mathcal{C}(\mathcal{E})\neq\emptyset$, i.e., if there is at least
one global state with those reductions. \ed

The following definition links the compatibility of reduced states to the correlation properties of global states.

\bd \label{def:compatible2} We say that the reductions
$\mathcal{E}$ are \emph{incompatible with a set $\mathcal{S}$} (or
with the defining correlation property of $\mathcal{S}$) if
$\mathcal{C}(\mathcal{E})\cap\mathcal{S}=\emptyset$. We say that (a
compatible) $\mathcal{E}$ is only compatible with genuine
multipartite entanglement if it is incompatible with
$\mathcal{S}_{\textrm{BS}}$. \ed

Figure~\ref{fig:sets} illustrates the problem of deciding whether certain bipartite reductions necessarily require the presence of global correlations of a certain kind.

We can always expand a generic tripartite state as
\begin{align}\label{eq:statecorrtens}
\r_{ABC}=& \r_A \otimes \r_B \otimes \r_C + \chi_{ABC} \nonumber\\
&+ \chi_{AB} \otimes \frac{\openone_C}{d_C} + \chi_{AC} \ox \frac{\openone_B}{d_B}
+  \chi_{BC} \ox \frac{\openone_A}{d_A}
\end{align}
where $\r_k$ is the reduced states of party $k$ and $\openone_k/d_k$ is the normalized identity operator on the Hilbert space $\mathcal{H}_k$. The bipartite correlation matrices $\chi_{kl}$ can be defined via $\r_{jk} = \rho_j \otimes \rho_k + \chi_{jk}$, and satisfy ${\tr}_k[\chi_{kl}] = {\tr}_l [\chi_{kl}]=0$. It is worth noting that when a bipartite marginal state $\r_{jk}$ is fully classical, then $[\r_j \otimes \r_k , \chi_{jk}] = 0$, and $\chi_{jk}$ is also diagonal in the same product basis as $\rho_{jk}$. The tripartite correlation matrix $\chi_{ABC}$, which for a fixed $\r_{ABC}$ can be defined via \eqref{eq:statecorrtens}, satisfies ${\tr}_k[\chi_{ABC}] = 0$, for all $k\in\{A,B,C\}$. For compatible reductions $\mathcal{E} = (\r_{AB},\r_{AC},\r_{BC})$, the compatibility set $\mathcal{C}(\mathcal{E})$ is spanned by choosing the tripartite correlation matrix $\chi_{ABC}$ so that the resulting operator in \eqref{eq:statecorrtens} is a physical state, i.e., positive semidefinite. On the other hand, to determine whether a triple of bipartite states is compatible, we first check the basic necessary condition that the single party marginals be the same, i.e., $\tr_j[\rho_{ij}] = \tr_k[\rho_{ik}]$ for all $\{i,j,k\} \in \{A,B,C\}$. Next, we have to  search for a tripartite correlation matrix,
$\chi_{ABC}$, such that Eq.~\eqref{eq:statecorrtens} is physical. If no such $\chi_{ABC}$ exists,  the given states are not compatible.

\section{Classicality of reductions and global entanglement}
\label{sec:classicalcomp}

Consider the
marginals from the well-known Greenberger-Horne-Zeilinger (GHZ)
state $\ket{\textrm{GHZ}}= (\ket{000} + \ket{111})/\sqrt{2}$:
$\r_{AB} = \r_{BC} = \r_{AC} =\frac{1}{2} (\proj{00} + \proj{11})$.
These are fully classical marginals coming from a genuinely
tripartite entangled state. However, these marginals are also
compatible with $\frac{1}{2} (\proj{000} + \proj{111})$, which is
fully-classical.

In this section we will first provide an example where fully classical two-body reduced states are not compatible with a global fully classical state, and actually require the presence of entanglement. On the other hand, we will prove that, in the case of three qubits, the fully classical two-body reductions  are always compatible with a global states that is not genuine multipartite entangled.

\subsection{Two-body classical states may require global quantumness of correlations}

 We will now derive conditions to ensure that
some fully-classical marginals cannot be compatible with any global
fully-classical state. We start with the following lemma.

\bl \label{le:commutatorlemma} Suppose three states $\r_{AB}$,
$\r_{AC}$, and $\r_{BC}$ are fully classical. Consider the
commutators $\Delta_{ij,ik}=\left[ \chi_{ij} \ox \openone_k ,\;
\chi_{ik} \ox \openone_j \right] = \left[ \rho_{ij} \ox \openone_k
,\; \rho_{ik} \ox \openone_j \right]$, where the second equality is
due to the assumed classicality, i.e., $[\r_j \otimes \r_k ,
\chi_{jk}] =0$. Then: 
\begin{enumerate}[(i)]
\item All commutators $\Delta_{ij,ik}$ vanish
if and only if there are orthonormal basis $\{\ket{a_i}\}$,
$\{\ket{b_i}\}$, $\{\ket{c_i}\}$, such that $\r_{AB} = \sum_{ij}
p_{ij}\proj{a_i,b_j}$, $\r_{BC} = \sum_{ij} q_{ij}\proj{b_i,c_j}$,
and $\r_{AC} = \sum_{ij} r_{ij}\proj{a_i,c_j}$.
\item If some
commutator $\Delta_{ij,ik}$ does not vanish, then:
\begin{enumerate}[(a)]
\item at least one
$\r_i$ is degenerate (i.e., at least two eigenvalues of some $\r_i$
are the same);
\item there does not exist a tripartite fully classical
state that is compatible with $\r_{AB}$, $\r_{BC}$, and $\r_{AC}$.
\end{enumerate}
\end{enumerate}
\el


\bpf (i) The ``if'' part is trivial. Let us prove the ``only if''
part. By hypothesis we may assume
 \begin{align}
\r_{AB} =& \sum_{ij}p_{ij}\proj{a_i,b_j}  \label{eq:appendixAB} \notag\\
=& \sum_{j}p_{j}\a_j\ox\proj{b_j} = \sum_{i}p_{i}'\proj{a_i}\ox\b_i, \\
\r_{AC} =& \sum_{ij}r_{ij}\proj{a_i',c_j} \label{eq:appendixAC} \notag\\
=& \sum_{j}r_{j}\a_j'\ox\proj{c_j} = \sum_{i}r_{i}'\proj{a_i'}\ox\g_i, \\
\r_{BC} =& \sum_{ij}q_{ij}\proj{b_i',c_j'} \label{eq:appendixBC} \notag\\
=& \sum_{j}q_{j}\b_j'\ox\proj{c_j'} = \sum_{i}q_{i}'\proj{b_i'}\ox\g_i',
\end{align}
with the orthonormal basis $\{\ket{a_i}\},\{\ket{a_i'}\}$ on $\cH_A$,
$\{\ket{b_i}\},\{\ket{b_i'}\}$ on $\cH_B$,
$\{\ket{c_i}\},\{\ket{c_i'}\}$ on $\cH_C$, and $\alpha_j = \sum_i p_{ij} \proj{a_i} / p_j$, $p_j=\sum_i p_{ij}$ (similarly for $\beta_j$, etc). Then $\Delta_{ij,ik} = 0$ implies
 \begin{align}
&\left[ \c_{ij} \ox \openone_k, \; \c_{ik} \ox \openone_j \right]
\label{eq:commute} \notag\\
& \quad=\left[ (\r_{ij}-\r_i\ox\r_j) \ox \openone_k, \; (\r_{ik}-\r_i\ox\r_k) \ox \openone_j \right] \notag\\
& \quad=\left[ \r_{ij} \ox \openone_k, \; \r_{ik} \ox \openone_j \right]=0,
\end{align}
with $i,j,k\in\{A,B,C\}$. By setting $i=A$ in \eqref{eq:commute}, we have $[\a_s,\a_t']=0$, $\forall \; s,t$. Thus the states $\a_s,\a_t'$ are simultaneously diagonalizable in the orthonormal basis $\{\ket{a_i''}\}$. So we may replace the basis $\{\ket{a_i}\}$ and $\{\ket{a_i'}\}$ in \eqref{eq:appendixAB} and \eqref{eq:appendixAC} by $\{\ket{a_i''}\}$. This replacement may result in the change of $p_i',\b_i$ and $r_i',\g_i$. Since there is no confusion, we still use them in \eqref{eq:appendixAB} and \eqref{eq:appendixAC}.

Next by setting $i=B$ in \eqref{eq:commute}, we have $[\b_s,\b_t']=0$, $\forall \; s,t$. Thus the states $\b_s,\b_t'$ are simultaneously diagonalizable in the orthonormal basis $\{\ket{b_i''}\}$. So we may replace the basis $\{\ket{b_i}\}$ and $\{\ket{b_i'}\}$ in \eqref{eq:appendixAB} and \eqref{eq:appendixBC} by $\{\ket{b_i''}\}$. This replacement may result in the change of $q_i',\g_i'$. Since there is no confusion, we still use them in \eqref{eq:appendixBC}. Third we set $i=C$ in \eqref{eq:commute} and repeat the above argument to show that the basis $\{\ket{c_i}\}$ and $\{\ket{c_i'}\}$ in \eqref{eq:appendixAC} and \eqref{eq:appendixBC} can be replaced by the orthonormal basis $\{\ket{c_i''}\}$. So the assertion follows.

(ii) Suppose either condition (a) or (b) is violated. We have that  either $\r_A$, $\r_B$, and $\r_C$ are all non-degenerate respectively in the orthonormal basis $\{\ket{a_i}\}$, $\{\ket{b_i}\}$, $\{\ket{c_i}\}$ (violation of (a)), or that there exists a tripartite fully classical state $\sum_{i,j,k}f_{ijk}\proj{a_i,b_j,c_k}$ that is compatible with $\r_{AB}$, $\r_{BC}$, and $\r_{AC}$ (violation of (b)). In either case we have $\r_{AB} = \sum_{ij} p_{ij}\proj{a_i,b_j}$, $\r_{AC} = \sum_{ij} r_{ij} \proj{a_i,c_j}$ and $\r_{BC} = \sum_{ij} q_{ij}\proj{b_i,c_j}$. So (i) implies that all commutators $\Delta_{ij,kl}$ vanish, and we have a contradiction.
\epf


An immediate consequence of Lemma~\ref{le:commutatorlemma} is the following.

\begin{theorem}
\label{thm:commuteimpliesclassical}
Let $\r_{AB}$, $\r_{BC}$, and $\r_{AC}$ be three
compatible bipartite fully classical states, such that

\smallskip\noindent (i) they all commute (all commutators $\Delta_{ij,kl}$ of Lemma~\ref{le:commutatorlemma} vanish),

\smallskip\noindent or

\smallskip\noindent (ii) all three one-body reduced states $\r_A$, $\r_{B}$ and $\r_C$ are
non-degenerate.

\smallskip\noindent Then $\r_{AB}$, $\r_{BC}$, and $\r_{AC}$ are
compatible with a fully classical tripartite state.
\end{theorem}
\begin{proof} The fact that all three single-system reductions are not degenerate implies, by Lemma~\ref{le:commutatorlemma}.(iia), that all the commutators $\Delta_{ij,kl}$ defined in Lemma~\ref{le:commutatorlemma} vanish. By Lemma~\ref{le:commutatorlemma}.(i) we have that $\r_A$, $\r_B$, and $\r_C$, are diagonal
in the orthonormal bases $\{\ket{a_i}\}$, $\{\ket{b_j}\}$,
and $\{\ket{c_k}\}$, respectively, in which $\r_{AB}$, $\r_{BC}$, and $\r_{AC}$ are explicitly classical.
Most importantly, we have
 \begin{gather}
\r_{XY} = \sum_{ij} \proj{x_i y_j} \r_{XY} \proj{x_i y_j},
 \end{gather}
with $x,y \in \{a,b,c\}$ and $X,Y \in \{A,B,C\}$. Let $\r_{ABC}$ be any tripartite state with which the three two-body reductions are compatible. Then also the fully classical tripartite state
\begin{gather}
\s_{ABC} = \sum_{ijk} \proj{a_i b_j c_k} \r_{ABC} \proj{a_i b_j c_k}
\end{gather}
has bipartite reduced density matrices $\r_{AB}$, $\r_{BC}$, and $\r_{AC}$.
\end{proof}


Given Theorem \ref{thm:commuteimpliesclassical}, in order to construct an example where $\r_{AB}$, $\r_{AC}$,
and $\r_{BC}$ are all fully classical but not compatible with any fully-classical state, we have first of all to construct an example where   $\r_{AB}$, $\r_{AC}$,
and $\r_{BC}$ are classical but do not commute with each
other. For this,  we will need the following lemma. We recall that for a bipartite state $\r$ acting on the Hilbert space $\cH_A \ox
\cH_B$, the partial transpose computed in the standard orthonormal
basis $\{\ket{i}\}$ of system $A$, is defined by $\r^\G
=\sum_{ij}\ketbra{j}{i}\ox\bra{i}\r\ket{j}$. One can similarly
define the partial transpose $\G_{B}$ on the system $B$.

\bl \label{le:classicalMAIN}
Consider three classical-classical two-qubit states
\begin{widetext}
\begin{align} \label{eq:classical2qubit1}
\r_{AB} =& p(\proj{00}+\proj{11}) + (1/2-p)(\proj{01}+\proj{10}),\\
\label{eq:classical2qubit2}
\r_{BC} =& q(\proj{b_0,0}+\proj{b_1,1}) + (1/2-q)(\proj{b_0,1}+\proj{b_1,0}),\\
\label{eq:classical2qubit3}
\r_{AC} =& r(\proj{a_0,c_0} + \proj{a_1,c_1}) + (1/2-r) (\proj{a_0,c_1} +
\proj{a_1,c_0}),
\end{align}
\end{widetext}
where $p,q,r\in(0,1/4)$, and any one of $\{\ket{a_i}\}$, $\{\ket{b_i}\}$,
$\{\ket{c_i}\}$ is a real and orthonormal basis in $\bC^2$. Let
\begin{align}
\r_{ABC} =& -\frac14 \openone_A \otimes \openone_B \otimes \openone_C \notag\\
&+ \r_{AB} \ox \frac{\openone_C}{2} + \r_{AC} \ox \frac{\openone_B}{2} + \r_{BC}
\ox \frac{\openone_A}{2}.\label{eq:classicalMAIN}
\end{align}
Then:
\begin{enumerate}[(i)]
\item If $\r_{ABC}\ge0$ then $\r_{ABC}$ is separable with respect to to the partition $A:BC$, $B:AC$ and $C:AB$.
\item $\r_{AB}$, $\r_{BC}$ and $\r_{AC}$ are compatible if and only if they are compatible with the biseparable state $\r_{ABC}$ in \eqref{eq:classicalMAIN}.
\item Suppose $\r_{AB}$, $\r_{BC}$ and $\r_{AC}$ are compatible. They are compatible with a fully separable state if and only if $\r_{ABC}$ is fully separable.
\end{enumerate}
\el
\bpf
(i) One may directly verify that the state is invariant under partial transposition with respect to any system, i.e., $\r^{\G_X}=\r$ for $X=A,B,C$. Since  $\r_{ABC}\ge0$, the assertion follows from~\cite[Theorem 2]{kck00}.

(ii) The ``if'' part is trivial; let us prove the ``only if'' part.
Suppose the bipartite marginals $\r_{AB}$, $\r_{BC}$ and $\r_{AC}$
are compatible with a tripartite state $\r'_{ABC}$. Since
$\r_{AB}$, $\r_{BC}$, and $\r_{AC}$ are real, they are also
compatible with the real state $(\rho'_{ABC}+\rho'^*_{ABC})/2$, so
we can assume that $\rho'_{ABC}$ is real without loss of generality.
By Eq.~(1), there is a Hermitian matrix $\chi_{ABC}$
such that
\[
\r'_{ABC}=\r_{ABC}+\chi_{ABC}.
\]
Since in our case both $\rho$ and $\rho'$ are real, also $\chi$ is real.
It follows from Eqs.~\eqref{eq:classical2qubit1}--\eqref{eq:classical2qubit3}, and the fact that $\{\ket{a_i}\}$, $\{\ket{b_i}\}$,
$\{\ket{c_i}\}$ are real and orthonormal bases, that $\r_{AB}$, $\r_{BC}$ and $\r_{AC}$ are invariant under the local unitary $\s_y\ox\s_y$. So they are compatible with the state
\begin{equation}
\label{eq:compatiblemix}
\begin{split}
&\frac{1}{2}\bigg(\r'_{ABC}+(\s_y\ox\s_y\ox\s_y) \r'_{ABC} (\s_y\ox\s_y\ox\s_y) \bigg)\\
=&\frac{1}{2}\bigg(\r_{ABC}+(\s_y\ox\s_y\ox\s_y) \r_{ABC} (\s_y\ox\s_y\ox\s_y) \bigg)\\
 &+\frac{1}{2}\bigg(\chi_{ABC}+(\s_y\ox\s_y\ox\s_y) \chi_{ABC} (\s_y\ox\s_y\ox\s_y)  \bigg)\\
 =& \r_{ABC}\\
 &+\frac{1}{2}\bigg(\chi_{ABC}+(\s_y\ox\s_y\ox\s_y) \chi_{ABC} (\s_y\ox\s_y\ox\s_y)  \bigg),
\end{split}
\end{equation}
where we have used that, from~\eqref{eq:classicalMAIN},
\begin{gather}
(\s_y\ox\s_y\ox\s_y)\r_{ABC}(\s_y\ox\s_y\ox\s_y) = \r_{ABC} .\notag
\end{gather}
We will now argue that, for a real $\chi$,
\begin{equation}
\label{eq:vanishingchi}
\chi_{ABC}+(\s_y\ox\s_y\ox\s_y) \chi_{ABC} (\s_y\ox\s_y\ox\s_y)=0,
\end{equation}
so that \eqref{eq:compatiblemix} proves that, for a physical state $\rho'$, $\rho$ is also physical, as it corresponds to the convex combination of physical states.
The starting point in proving \eqref{eq:vanishingchi} is to observe that every three-qubit correlations matrix $\chi$ is by definition the linear combination of traceless Pauli matrices, i.e.
\begin{equation}
\label{eq:chi}
\chi = \sum_{i,j,k=1}^3 \chi_{ijk} \sigma_i \ot \sigma_j \ot \sigma_k.
\end{equation}
Since $\chi$ is Hermitian, all coefficients $ \chi_{ijk} $ are real. Moreover, for a real (and hence symmetric) $\chi$ only terms with an even number of $\sigma_2=\sigma_y$ are present in the expansion, because $(\sigma_2)^T=-\sigma_2$, while $\sigma_1=\sigma_x$ and $\sigma_3=\sigma_z$ are symmetric. On the other hand, $\sigma_2 \sigma_m \sigma_2 = -\sigma_m$ for $m=1,3$, while, obviously, $\sigma_2 \sigma_2 \sigma_2 = \sigma_2$. Since each non-zero term in the expansion~\eqref{eq:chi} of $\chi$ contains an even number of $\sigma_2$'s, it will change sign after conjugation by $\sigma_2\ot\sigma_2\ot\sigma_2$, i.e.,
\[
\begin{split}
&(\s_y\ox\s_y\ox\s_y) \chi_{ABC} (\s_y\ox\s_y\ox\s_y)\\
=&\sum_{i,j,k=1}^3 \chi_{ijk} (\s_y\ox\s_y\ox\s_y) \sigma_i \ot \sigma_j \ot \sigma_k(\s_y\ox\s_y\ox\s_y)\\
=&\sum_{i,j,k=1}^3 \chi_{ijk} (\s_y\sigma_i \s_y)\ot (\s_y\sigma_j \s_y)\ot (\s_y\sigma_k\s_y)\\
=& - \sum_{i,j,k=1}^3 \chi_{ijk} \sigma_i \ot \sigma_j \ot \sigma_k\\
=& - \chi_{ABC}.
\end{split}
\]
As argued, this implies $\r_{ABC}\ge0$, with biseparability following from (i).

(iii) The ``if'' part follows from (ii), let us prove the ``only if'' part. Suppose $\r_{AB}$, $\r_{BC}$ and $\r_{AC}$ are compatible with a fully separable state $\r'_{ABC}$. From \eqref{eq:compatiblemix}, $\r_{ABC}$ is the convex sum of a few fully separable states. So the assertion follows. This completes the proof.
\epf

We are now ready to present our example where $\r_{AB}$, $\r_{AC}$,
and $\r_{BC}$ are all fully classical but not compatible with any fully-classical state.

\bex \label{eg:commutatorex}
\rm{Consider the three-qubit state
\begin{align}
 \r_{ABC}(q) =& \frac{1}{8} (\openone_A \ox \openone_B \ox \openone_C + q \;
\openone_A \ox \sigma_1 \ox \sigma_1 \nonumber\\ &+ q \; \sigma_2
\ox \openone_B \ox \sigma_2 + q \; \sigma_3 \ox \sigma_3 \ox
\openone_C),
 \label{ea:nondegenerate4}
\end{align}
where $\s_i$, with $i=1,2,3$ are the Pauli matrices and
$\frac{-1}{\sqrt{3}} \le q \le \frac{1}{\sqrt{3}}$ (for $q$ outside
of this interval the matrix $\r_{ABC}(q)$ is not positive
semi-definite), $q\neq 0$. It is not hard to see that each of the
bipartite marginals states is fully classical but with respect to
different bases. A quick method to verify this assertion is using
Theorem 1 in~\cite{ccm11}. Moreover the reductions of $\r_{ABC}(q)$
do not commute with each other. Part (ii) of
Lemma~\ref{le:commutatorlemma} implies that there is no tripartite
fully classical state that is compatible with these bipartite
marginals: $\mathcal{C} (\r_{ABC}(q)) \cap \mathcal{S}_\textrm{FS} =
\emptyset$.
Surprisingly, we can find a value of  $q$ for which the fully
classical marginals in fact require some global entanglement.
Consider the state $\o_{ABC} = \r_{ABC} (q=1/\sqrt{3})$. The range
of $\o_{ABC}$ does not contain any fully factorized pure state,
hence it cannot be fully separable. Now, up to a local unitary,
the classical bipartite marginals $\o_{AB}$, $\o_{AC}$, $\o_{BC}$,
and the entangled state $\o_{ABC}$ can be written
as~\eqref{eq:classical2qubit1}-\eqref{eq:classicalMAIN},
respectively. It follows from Lemma \ref{le:classicalMAIN}.(iii)
that $\o_{AB}$, $\o_{AC}$, and $\o_{BC}$ cannot be compatible with
any fully separable state.
%
%
%
}
\eex

In Appendix~\ref{sec:appuni}, which focuses on the uniqueness of global states with fixed two-body reductions, we provide also an example of classical reductions  compatible with a \emph{unique} global state that is not fully classical, although fully separable.

\subsection{Classical two-qubit states do not require genuine tripartite entanglement} 

Although, as we have just seen, compatible classical marginals may require global quantum
correlations or even entanglement, it turns out that for the case of 
\emph{three qubits} they will never require the global state to be
genuinely multipartite entangled.

To see this, will need an additional lemma.

\bl \label{le:nongenuine}
Let $\{\ket{a_i}\}$, $\{\ket{b_i}\}$ be an orthonormal basis on $\cH_A$ and $\cH_B$, respectively. Suppose the bipartite marginals $\r_{AB}$, $\r_{AC}$ and $\r_{BC}$ are compatible. They are compatible to a non-genuinely entangled tripartite state when one of the following conditions is satisfied: (i) $\r_{AB}=\sum_i p_i \proj{a_i}\ox\r_i$ and $\r_{AC}=\sum_i q_j \proj{a_i}\ox\s_i$.  (ii) $\r_{AB}=\sum_i r_i \proj{a_i,b_i}$.
\el

\begin{proof}
Suppose $\r_{AB}$, $\r_{AC}$ and $\r_{BC}$ are compatible with $\r_{ABC}$. If hypothesis (i) is satisfied, then they are also compatible with the state $\sum_i \proj{a_i}_A \ox \openone_{BC} \r_{ABC}\proj{a_i}_A \ox \openone_{BC}$, which is biseparable. On the other hand if hypothesis (ii) is satisfied, then it follows from~\cite{cw04} that there is a quantum channel $\L$ on $\cH_C$ such that $\r_{ABC}=\L(\proj{\psi})$, for $\ket{\psi}=\sum_i \sqrt{r_i} \ket{a_i,b_i,i}$. So we obtain $\r_{AC}=\L(\sum_i r_i \proj{a_i,i})$ and $\r_{BC}=\L(\sum_i r_i \proj{b_i,i})$. Then $\r_{AB}$, $\r_{AC}$ and $\r_{BC}$ are compatible with the fully separable state $\L(\sum_i r_i \proj{a_i,b_i,i})$. This completes the proof.
\end{proof}

Now we are ready to give the proof of the following.

\bt \label{thm:classical}
Any three compatible classical-classical two-qubit states $\r_{AB}$, $\r_{BC}$ and $\r_{AC}$ are compatible with a tripartite biseparable state.
\et

\begin{proof}
Suppose some fully classical
bipartite marginals $\r_{AB}$, $\r_{BC}$ and $\r_{AC}$ are only
compatible with genuinely entangled states $\r_{ABC}$. Then the
one-party reduced density operators $\r_A,\r_B,\r_C$ have to be the
maximally mixed states, $\openone/2$. Indeed, if, without loss of
generality in the argument, $\rho_A$ is non-degenerate, the
two-party reduced states would also be compatible with a global
state given by the a locally (on $A$) dephased version of
$\r_{ABC}$, which would be separable in $A:BC$, leading to a
contradiction. Thus, up to local unitaries, we may assume $\r_{AB} =
p\proj{00} + x\proj{01} + y\proj{10}+z\proj{11}$ where $p+x+y+z=1$
and $0\le p\le 1/4$. Since $\r_A=\r_B=\r_C=\openone/2$, we have
$p+x=p+y=y+z=1/2$. So we obtain $x=y,p=z\in[0,1/4]$. Since
$\r_{ABC}$ is genuinely entangled, the cases $p=1/4$ and $p=0$ are
excluded by Lemma~\ref{le:nongenuine}.(i) and~\ref{le:nongenuine}.(ii), respectively. So
we obtain $\r_{AB}$ as Eq.~\eqref{eq:classical2qubit1}. By similar
arguments and performing suitable diagonal local unitary gates on
systems $A,B$, the classical-classical two-qubit states $\r_{BC}$
and $\r_{AC}$ can be simplified to the forms
Eqs.~\eqref{eq:classical2qubit2} and \eqref{eq:classical2qubit3},
respectively. Meantime, $\r_{AB}$ is unchanged. Since $\r_{AB}$,
$\r_{BC}$ and $\r_{AC}$ are compatible, it follows from
Lemma~\ref{le:classicalMAIN}.(ii) that they are compatible with a
biseparable state. It gives us a contradiction. So there are no
compatible bipartite marginals that are only compatible with genuine
entangled states. This completes the proof.
\end{proof}

\section{Separable reductions can imply genuine multipartite
entanglement}
\label{sec:genuine}

We have seen that the condition of classicality of marginals is strong enough to exclude the need for global genuine multipartite entanglement. We will now construct non-classical separable marginals that are only compatible with global genuine multipartite entanglement, but
first we need to establish some more definitions and notation.

We set $d_A=\dim\cH_A$, $d_B=\dim\cH_B$ and $d_C=\dim\cH_C$.
We denote $r(M)$, $\cR(M)$ the rank and range of any square matrix $M$, respectively. A quantum state is a positive semidefinite linear operator $\r:\cH\to\cH$ with $\tr\r=1$. We say $\r_{ABC}$ is a $m\times n\times l$ state by meaning that the reduced density operators satisfy $r(\r_A)=m$, $r(\r_B)=n$, $r(\r_C)=l$.
The ranks of the reduced density operators of $\r_{ABC}$ are invariant when we perform an invertible local operator (ILO) on $\r_{ABC}$. That is, let $A=\bigox^3_{i=1} A_i \in \GL: = \GL_{d_A} (\bC) \times \GL_{d_B} (\bC) \times \GL_{d_C} (\bC)$ such that $\s=A\r A^\dg$. Then $r(\r_X)=r(\s_X)$, $r(\r_{XY})=r(\s_{XY})$ and $r(\r)=r(\s)$ where $X,Y=A,B,C$. We also denote $\ket{a^*}$ as the vector whose components are the complex conjugate of those of $\ket{a}$. So $\ket{a}$ is real when $\ket{a}=\ket{a^*}$.

Evidently
$r(\r^\G)=r(\r^{\G_B})$, where $\G$, we recall, denotes partial transposition. We call the integer pair $(r(\r),r(\r^\G))$
the \textit{birank} of $\r$, and the two integers may be different.
For such examples of two-qubit and qubit-qutrit separable states, we
refer the readers to~\cite[Table I, II]{cd12paper9}. Furthermore, we
say that $\r$ is a PPT [NPT] state if $\r^\G\ge0$ [$\r^\G$ has at
least one negative eigenvalue]. Evidently, a separable state must be
PPT. The converse is true only if $mn\le6$~\cite{peres96, hhh96}.

We say a bipartite state $\r_{AB}$ is \textit{$A$-finite} when for any subspace $H\su\cH_A$, $\dim H>1$ and any state $\ket{x}\in\cH_B$, it holds that $H\ox\ket{x}\not\su \cR(\r_{AB})$. In other words, $\r_{AB}$ is not $A$-finite when $\cR(\r_{AB})$ contains a 2-dimensional subspace spanned by $\ket{a_1,x},\ket{a_2,x}$ with some linearly independent states $\ket{a_1},\ket{a_2}$. So if $\r_{AB}$ is not $A$-finite, there must be infinitely many product states in $\cR(\r_{AB})$. 

Besides these notions and notation, we will need also the following lemma.

\bl \label{le:symmetricAB}
Suppose the bipartite marginals $\r_{AB}$, $\r_{BC}$ and $\r_{AC}$ are compatible with $\r_{ABC}$, and $\r_{AB}$ is $A$-finite and $B$-finite. Then $\r_{ABC}$ is either separable with respect to the partition $AB:C$, or genuinely multipartite entangled.
\el

\bpf Suppose $\r_{ABC}$ is biseparable, so $\r_{ABC} = p \a_{A:BC} + q \b_{B:AC} + (1-p-q) \g_{C:AB}$, with $\alpha_{A:BC}$ separable in the $A:BC$ partition (similarly for $\beta_{B:AC}$ and $\gamma_{C:AB}$. We argue that $\a_{A:BC}$ is fully separable, and a similar argument will apply to $\b_{B:AC}$. Let $\a_{A:BC}=\sum_i p_i \proj{a_i}_A \ox \proj{\ps_i}_{BC}$. Since $\r_{AB}$ is $B$-finite and $\ket{a_i}\ox \cR(\tr_C \proj{\ps_i}) \su \cR(\r_{AB})$, any $\ket{\ps_i}$ must be a product state. So  $\a_{A:BC}$ is fully separable. Similarly one can show that $\b_{B:AC}$ is also fully separable, so $\r_{ABC}$ is separable with respect to $AB:C$.
\epf


We are now in the position to prove the following.

\bt \label{thm:main}
Suppose the triple $\mathcal{E}=(\r_{AB},\r_{BC},\r_{AC})$ is compatible. Then, $\mathcal{C}(\mathcal{E})\cap \mathcal{S}_\textup{BS}=\emptyset$  if all the following conditions are met: (i) for any $i,j\in\{A,B,C\}$, the state $\r_{ij}$ is $i$-finite and $j$-finite; (ii) $\r_{BC}$ has birank $(r(\r_B) + 1, r(\r_B)+1)$; (iii) $\r_{AB}$ has birank $(r,s)$, $r\ne s$.
\et

\begin{proof} Suppose $\r_{AB}$, $\r_{BC}$ and
$\r_{AC}$ are compatible with a biseparable state $\r_{ABC}$. By
hypothesis (i) and Lemma~\ref{le:symmetricAB}, we obtain that
$\r_{ABC}$ is separable with respect to the partition $AB:C$. Let
$\r_{ABC} = \sum^{n-1}_{i=0} p_i \r_i \ox \proj{c_i}$ where $p_i>0$,
the $\ket{c_i}$ on $\cH_C$ are pairwise linearly independent. Note
that the product subspace
$\cR((\r_i)_B)\ox\ket{c_i}\su\cR(\r_{BC})$, $\forall i$. This fact
and (i) imply that $r((\r_i)_B)=1$, $\forall i$. By similar
arguments we have $r((\r_i)_A)=1$. We may assume
$\r_i=\proj{a_i,b_i}$. Hence $\r_{ABC}=\sum^{n-1}_{i=0} p_i
\proj{a_i,b_i,c_i}$. By (ii) we have $n\ge d+1$ where $d=r(\r_B)$.
Without loss of generality, we may assume that the states
$\ket{b_i}$, $i=0,\cdots,d-1$ span $\cR(\r_B)$. We choose a suitable
ILO $V$ such that $ V\ket{b_i} \propto \ket{i}, i=0,\cdots,d-1$,
$V\ket{b_i} \propto \ket{f_i},i\ge d$ and $\ket{f_d}$ is real. By
performing $V$ on the state $\r_{ABC}$, we have
 \begin{align}
 \s_{ABC}
 =& (I\ox V\ox I)\r_{ABC}(I\ox V\ox I)^\dg
 \notag\\
 =& \sum^{d-1}_{i=0} q_i \proj{a_i,i,c_i} + \sum^{n-1}_{i=d} q_i
 \proj{a_i,f_i,c_i}, \notag
 \end{align}
and $q_i>0$ for any $i$. Since the operation $V$ does not change the rank of quantum states, it follows from (ii) that $\s_{BC}$ has birank $(d+1, d+1)$. Recall that the $\ket{c_i}$ are pairwise linearly independent. Since (i) is not changed under ILOs, it follows from (i) that $\ket{f_d}$ is not parallel to any state $\ket{i}$---otherwise $\rho_{BC}$ would not be $B$-finite. Since $\ket{f_d}$ is real, the two $(d+1)$-dimensional subspaces $\cR(\s_{BC})$ and $\cR(\s_{BC}^\G)$ are equal and spanned by $\ket{i,c_i},i = 0, \cdots,d-1$ and $\ket{f_d,c_d}$. So the states $\ket{f_i,c_i}, \ket{f_i^*,c_i} \in \cR(\s_{BC})$, for any $i>d$. Then (i) implies that these $\ket{f_i}$ are real up to an overall phase. So $\s_{AB} = \s_{AB}^{\G_B}$. It implies $r(\r_{AB})=r(\r_{AB}^\G)$ which is a contradiction with (iii).
Therefore $\r_{AB}$, $\r_{BC}$ and $\r_{AC}$ are not compatible to any non-genuinely entangled state. This completes the proof.
\end{proof}

We will now make use of Theorem~\ref{thm:main} and offer an example of separable marginals only compatible with genuinely entangled tripartite states.
\bex \label{ex:main}
\rm{Consider the family of  rank-$(d+1)$ states on $\mathbb{C}^{d+1}\otimes \mathbb{C}^{d+1}\otimes \mathbb{C}^{d+1}$ given by
$\r_{ABC} = p_1 \s_{ABC} + \sum^{d}_{m=2} p_m\proj{mmm}$,
with
\begin{gather} \label{ea:essentialTRI}
\s_{ABC} = \frac{2}{3} \ket{\xi}\bra{\xi} +\frac{1}{3} \proj{111},
\end{gather}
$\ket{\xi}=\frac12\ket{010} + \frac12\ket{100} + \frac{1}{\sqrt2} \ket{001}$, $p_1>0$, and $p_m\ge0$. It is easy to see that the only biseparable pure state in $\cR(\s_{ABC})$ is $\ket{111}$. The bipartite reduced density operators of $\s_{ABC}$ are
\begin{align}
& \s_{AB} = \frac{1}{3}\ket{\Phi^+}\bra{\Phi^+} + \frac13\proj{00} +
\frac13\proj{11}, \label{ea:rhoAB}\\
& \s_{BC} = \s_{AC} = \frac{1}{2} \ket{\zeta}\bra{\zeta}+ \frac16\proj{00} +
\frac13\proj{11}, \label{ea:rhoAC,BC}
\end{align}
with $\ket{\Phi^+}=\frac{1}{\sqrt{2}}(\ket{01} + \ket{10})$ and $\ket{\zeta} = \sqrt{\frac{2}{3}}\ket{01} + \sqrt{\frac{1}{3}} \ket{10}$.

The three two-qubit marginals $\s_{AB}$, $\s_{BC}$ and $\s_{AC}$ are positive under partial transposition (PPT), so they are separable~\cite{hhh96}. Hence,  $\r_{AB}$, $\r_{BC}$ and $\r_{AC}$ are separable too; they also evidently satisfy  condition (i) of Theorem~\ref{thm:main}. Furthermore, $\r_{BC}$ has birank $(d+2,d+2)$, while $r(\r_B)=d+1$, and $\r_{AB}$ has birank $(d+2,d+3)$. So also conditions (ii) and (iii) of Theorem~\ref{thm:main} are satisfied, and we conclude that $\r_{AB}$, $\r_{BC}$ and $\r_{AC}$ are only compatible with genuinely tripartite entangled states.}
\eex

The example shows that for any fixed local dimension $d$, there
exist triples of  two-qudit separable states that are only
compatible with genuine multipartite entanglement. The ``core'' of
our construction is the genuine multipartite entangled three-qubit
state $\sigma_{ABC}$ of Eq.~\eqref{ea:essentialTRI}.
It turns out that $\sigma_{ABC}$ is actually  the only state
compatible with its reductions. The proof of this is given in the Appendix.
It is worth comparing this with the results
of~\cite{LPW2002}. There it was proven that for almost all pure
entangled states of three qubits $\ket{\eta}$ it holds
$\cC(\proj{\eta}) = \{\proj{\eta}\}$, with the exception of states
of the generalized GHZ form
$\ket{\textrm{gGHZ}}=\sqrt{p}\ket{000}+\sqrt{1-p}\ket{111}$ (up to
local unitary transformations), which satisfy, e.g.,
$\{\proj{\textrm{gGHZ}} ,p \proj{000}+(1-p)
\proj{111}\}\subset\cC(\proj{\textrm{gGHZ}})$. Interestingly, the
only three-qubit pure states that have separable reduction are of
the generalized GHZ form~\cite{thapliyal99}. This implies that any
three-qubit state $\rho$ such that (i) its reductions are separable
and (ii) $\cC(\rho)\cap \cS_{\textrm{BS}} = \emptyset$, must be
mixed. Since the state $\sigma_{ABC}$  has rank two, we can think of
it as the simplest possible example that satisfies (i) and (ii),
with the additional property of being uniquely determined by its
reductions. We generalize Example~\ref{ex:main} in several ways, all
presented in the Appendix.

\subsection{Genuine multipartite entanglement from separable reductions is a robust feature}

While we showed that there exist genuine multipartite states whose compatibility set contains only genuine multipartite states,
 it is natural to ask how common this phenomenon is, i.e., whether such states have finite volume in the set of all states. This is important also from the point of view of the potential realization of such states in the lab, which can never be perfect. We answer this question in the affirmative.

We introduce a parameter of compatibility  of a tripartite state
$\rho_{ABC}$ with $\cE=(\s_{AB},\s_{BC},\s_{AC})$ as $
D(\rho_{ABC}|\cE) := \norm{\r_{AB} - \s_{AB}}^2_2+ \norm{\r_{BC} -
\s_{BC}}^2_2+ \norm{\r_{AC} - \s_{AC}}^2_2 $, where we have used the
Hilbert-Schmidt norm $\|X\|_2=\sqrt{\tr(X^\dagger X)}$~\footnote{We
make this choice for the sake of concreteness, but  our argument is
only based on continuity of $D$ in its arguments and the fact that
$D(\rho_{ABC}|\cE)$ is positive and vanishes if an only if
$\rho_{ABC}\in\cC(\cE)$.}. We further define
$D_{\textrm{BS}}(\cE):=\min_{\r \in \cS_{\textrm{BS}}} D(\rho_{ABC}|
\cE)$. We have $D_{\bisep} (\cE)>0$ for any triple $\cE$ such that
$\cC(\cE)\cap \cS_{\bisep}=\emptyset$, even if the triple of reduced
states is compatible, as in Example~\ref{ex:main}. Finally, given a
tripartite state $\sigma_{ABC}$, we define $D(\rho_{ABC}|\s_{ABC}):=
D(\rho_{ABC}|(\s_{AB},\s_{BC},\s_{AC}))$ and
$D_{\bisep}(\s_{ABC}):=\min_{\r \in
\cS_{\bisep}}D(\rho_{ABC}|\sigma_{ABC})$.

Now, 
consider a genuinely entangled multipartite state $\bar\sigma_{ABC}$ with separable reductions such that $D_{\bisep}(\bar\s_{ABC})>0$, and the convex combination of $\bar\sigma_{ABC}$ with an arbitrary  fully separable state $\rho^{\FS}$: $ \tau_p(\bar\sigma_{ABC}, \rho^{\text{\FS}}):=(1-p) \bar\sigma_{ABC}+p \rho^{\text{\FS}}, $ for $0\leq p\leq 1$. Since the set of biseparable states is closed, there exists $\bar{p}>0$ such that $\tau_p(\bar\sigma_{ABC},\rho^{\FS})$ is genuine multipartite entangled for all $\rho^{\FS}$ and all $0\leq p <\bar{p}$. Since $\rho^{\FS}$ is fully separable, so are the two-party reduced states of $\tau_p(\bar\sigma_{ABC},\rho^{\text{\FS}})$. Furthermore, since $D$ is continuous, there exists $\bar{p}_D>0$ such that $D_{\bisep} (\tau_p (\bar\sigma_{ABC}, \rho^{\FS}))>0$ for all $\rho^{\FS}$ and all $0\leq p <\bar{p}_D$.

%

For any local finite dimensions, the set of fully separable states has non-zero volume among all states, because there exists a ball of fully separable states around the maximally mixed state~\cite{gurvitsbarnum}. Thus, the argument above proves that also the set of tripartite states whose two-party marginals are separable but only compatible with genuine multipartite entanglement has non-zero volume.

\section{Conclusions}
\label{sec:conclusions}

We analyzed the relation between the character of correlations of tripartite states and the ones exhibited by their bipartite reductions, i.e., a version of the quantum marginal problem that focuses on the compatibility of bipartite reductions with certain global properties. We constructed examples where separable reductions are only compatible with genuine multipartite entanglement. Up to our knowledge, this is the ``largest'' known separation between the character of correlations of bipartite reductions and what can be inferred about the quality of correlations of the global state, based only on the knowledge of the reductions. On the other hand, at least for qubits we were able to prove that compatible reductions that are fully classical can always originate from a biseparable global state. Nonetheless, bipartite reductions that are fully classical may still require the presence of some entanglement in the global state. Our results show that the relation between global and ``local'' correlations is far from trivial. Notably, the notion of fully classical correlations is strong enough to ``break'' the need for genuine multipartite entanglement, but not the potential need for global entanglement altogether. An interesting open question is whether compatible classical-classical marginals in high dimension are always compatible with a biseparable tripartite state. Another question is how to quantitatively bound the certifiable genuine multipartite entanglement in terms of the non-classicality of the two-body reductions, at least in the three-qubit case. The latter problem is reminiscent of the case of entanglement distribution, where the non-classicality of correlations---rather than the entanglement---present between a quantum carrier and distant labs constitutes a bound on the entanglement that can be generated between the labs by exchanging the carrier~\cite{chuanetal2012, streltsovetal2012}.

\section*{Acknowledgments}

We are grateful to A. Ac{\`i}n and O.
G{\"u}hne for useful discussions. LC started this work when at the Institute for Quantum Computing \& Department of Pure Mathematics,
University of Waterloo, and was mainly supported by MITACS
and NSERC; LC's research was funded in part by the Singapore
National Research Foundation under NRF Award NRF-NRFF2013-01. OG is grateful for the support of the Austrian Science
Fund (FWF) and Marie Curie Actions (Erwin Schr\"odinger Stipendium
J3312-N27). KM was supported by the John Templeton Foundation,
National Research Foundation and the Ministry of Education
(Singapore) during the completion of this work. MP acknowledges support from NSERC, CIFAR, and Ontario
Centres of Excellence.

\appendix

\section{On the uniqueness of global states compatible with given reductions}
\label{sec:appuni}

We first prove that $\sigma_{ABC}$ in Eq.~\eqref{ea:essentialTRI} in the
paper is the only state compatible with its reductions, a fact of
interest in its own.
\bpp \label{le:ENTunique} For $\sigma_{ABC}$ in Eq.~\eqref{ea:essentialTRI} it holds
$\mathcal{C}(\sigma_{ABC})=\{\sigma_{ABC}\}$. \epp

\bpf  Suppose $\r=\r_{ABC}$ has the same reductions as
$\sigma_{ABC}$, i.e., $\rho\in\mathcal{C}(\sigma_{ABC})$. We can
always write its spectral decomposition as $\r = \sum^7_{i=0} p_i
\proj{\ps_i}$, where $\ket{\ps_i} = \sqrt{q_i} \ket{0,\a_i} +
\sqrt{1-q_i}\ket{1,\ph_i}$, and $\ket{\a_i}, \ket{\ph_i} \in \cH_B
\ox \cH_C$, $\forall i$. We have $\s_{BC} = \sum_i p_i (q_i
\proj{\a_i} + (1-q_i) \proj{\ph_i})$. It follows from Eq.~(5) in the
paper that $r(\s_{BC})=3$. So any four states of
$\ket{\a_i},i=0,\cdots,7$ are linearly dependent. Using the freedom
in the choice of the pure-state ensemble representation of a mixed
states~\cite{hughston1993}, we can choose a suitable linear
combination of $\ket{\ps_i}$, $i=0,1,2,3$, such that it is equal to
$\ket{1}_A\ket{\ph_3'}_{BC}$. So the state can be written as $\r =
\sum^3_{i=0} r_i \proj{\ps_i'}+\sum^7_{i=4} p_i \proj{\ps_i}$ where
$\ket{\ps_3'}=\ket{1}_A\ket{\ph_3'}_{BC}$. By applying this
procedure to another four states
$\ket{\ps_0'},\ket{\ps_1'},\ket{\ps_2'},\ket{\ps_j}$ with
$j=4,5,6,7$ respectively, we can realize
$\ket{\ps_j}=\ket{1}_A\ket{\ph_j'}_{BC}$.

By relabeling the states, we can write $\r = \sum^2_{i=0} p_i'
\proj{\ps_i} + p_3'\proj{1}\ox\r_0$ with $\r_0$ on $\cH_B\ox\cH_C$.
We have $\cR(\proj{1}\ox(\r_0)_B)\su\cR(\s_{AB})$ and
$\ket{11}\in\cR(\s_{AB})$ by Eq.~\eqref{ea:rhoAB}. Since $\s_{AB}$
is X-finite for $X=A,B$, we have $(\r_0)_B=\proj{1}$. By the similar
argument we can show $(\r_0)_C=\proj{1}$. So we have $\r =
\sum^2_{i=0} p_i' \proj{\ps_i} + p_3' \proj{111}$.

Let $\ket{\ps_i}=\sum^1_{j,k,l=0}c_{i,m}\ket{jkl}$ where $i=0,1,2$
and $m = 4j + 2k + l$. By Eq.~\eqref{ea:rhoAB} we have
$c_{i2}=c_{i4}$, $c_{i3} = c_{i5}$. By Eq.~\eqref{ea:rhoAC,BC}, we have
$c_{i2}=c_{i1}/\sqrt2$ and $c_{i6} = c_{i5}/\sqrt2$. These equations
imply for $i=0,1,2$, we have
\begin{widetext}
\begin{gather}
\ket{\ps_i} = c_{i0}\ket{000} + c_{i1} \bigg( \ket{001}
+\frac{1}{\sqrt2}\ket{010} + \frac{1}{\sqrt2}\ket{100} \bigg)
+ c_{i3} \bigg( \ket{011} + \ket{101} + \frac{1}{\sqrt2}\ket{110} \bigg)
+ c_{i7} \ket{111}.
 \end{gather}
The coefficients of $\proj{11}$ in both $\s_{AB}$ and $\s_{AC}$ are
$1/3$, so $c_{03}=c_{13}=c_{23}=0$. By replacing
$\ket{\ps_i},i=0,1,2$ by a suitable linear combination of them, we
may assume $c_{11}=c_{21}=c_{20}=0$. So the tripartite state can be
rewritten as $\r = \sum^2_{i=0} p_i'' \proj{\ps_i}$, where
\begin{gather}
\ket{\ps_0} = c_{00}'\ket{000} + c_{01}' \bigg( \ket{001}
+\frac{1}{\sqrt2} \ket{010} + \frac{1}{\sqrt2}\ket{100} \bigg) +
c_{07}'\ket{111}, \quad \ket{\ps_1} = c_{10}'\ket{000}+
c_{17}'\ket{111}, \quad \ket{\ps_2} = \ket{111}.
\end{gather}
\end{widetext}
Since $r(\s_{BC})=3$, we have $c_{01}'\ne0$. By Eq.~\eqref{ea:rhoAB}, we have $c_{01}'c_{00}'=c_{01}'c_{07}'=0$. So $c_{00}' =
c_{07}' =0$. By Eq.~\eqref{ea:rhoAB} again, we have $p_0''=\frac23$,
$\abs{c_{01}'}=\frac{1}{\sqrt2}$, and $c_{10}'=0$. Now we see
$\r=\s_{ABC}$ in Eq.~\eqref{ea:essentialTRI}. This completes the proof.
\epf

We further derive (and later use in Appendix \ref{sec:appgen}) the following lemma.

\bl \label{le:unique} Suppose $\r_{AB}$, $\r_{BC}$, $\r_{AC}$ are only compatible with a tripartite state $\r_{ABC}$, and $\s_{AB}$, $\s_{BC}$, $\s_{AC}$ are compatible with another tripartite state $\s_{ABC}$. If $\cR(\s_{ABC}) \sue \cR(\r_{ABC})$, then $\s_{ABC}$ is the only state with which $\s_{AB}$, $\s_{BC}$, $\s_{AC}$ are compatible.
\el

\bpf Suppose $\s_{AB},\s_{BC},\s_{AC}$ are compatible with another state $\s'_{ABC}\ne\s_{ABC}$. Since $\cR(\s_{ABC}) \sue \cR (\r_{ABC})$, we may find a small enough $p>0$ and a tripartite state $\a_{ABC}$, such that $\r_{ABC} = p\s_{ABC} + (1-p)\a_{ABC}$. So the bipartite reductions $\r_{AB}, \r_{BC}, \r_{AC}$ are compatible with the state $p\s_{ABC} '+ (1-p) \a_{ABC}$, which is different from $\r_{ABC}$. It gives us a contradiction.
\epf

We conclude this section presenting separable marginals that are only compatible with a unique  quantum correlated (unentangled) state.

\bpp \label{pp:SEPunique}
The separable states $\r_{AB}= \r_{BC} = \r_{AC} = p \proj{00} + (1-p) \proj{a,a}$, where $\ket{a} = \frac{1}{\sqrt2} (\ket{0}+\ket{1})$, are only compatible with the separable state $\r_{ABC} = p \proj{000} + (1-p)\proj{a,a,a}$.
\epp

\bpf
We will use the following observation in the proof and it is easy to verify. For any $X,Y\in\{A,B,C\}$, there are only two product states $\ket{00}, \ket{a,a} \in \cR(\r_{XY})$, they also span the space $\cR(\r_{XY})$. That is, any state in $\cR(\r_{XY})$ is the linear combination of $\ket{00}$ and $\ket{a,a}$.

It is clear that $\r_{AB},\r_{BC},\r_{AC}$ are compatible with $\r_{ABC}$. Suppose they are compatible with another three-qubit state $\s_{ABC} = \sum_{i} p_i \proj{\ps_i}$. By applying the observation to system $B,C$ we have $\ket{\ps_i} = f_i \ket{\a_i,00} + g_i\ket{\b_i,a,a}$ with some complex numbers $f_i,g_i$. By applying the observation to system $A,B$ we have $g_i\ket{\b_i}\propto\ket{a}$, and hence $f_i\ket{\a_i}\propto\ket{0}$. So we may assume $\ket{\ps_i} = f_i' \ket{000} + g_i'\ket{a,a,a}$. As a result, the range of the state $\s_{ABC}$ is spanned by the product states $\ket{000},\ket{a,a,a}$. By simple algebra one can see that the only feasible $\s_{ABC}$ compatible to $\r_{AB},\r_{BC}$ and $\r_{AC}$ is the convex sum of $\proj{000}$ and $\proj{a,a,a}$. By using the condition $\r_{XY}=\s_{XY}$ we obtain $\s_{ABC}=\r_{ABC}$. This completes the proof.
\epf

\section{Generalizations of Example~\ref{ex:main}}
\label{sec:appgen}

We provide here some further examples of states with separable
reductions that are only compatible with genuine multipartite
entanglement, also making use of Proposition~\ref{le:ENTunique} and Lemma
\ref{le:unique}.

Note that $\ket{111}\in\cR(\s_{ABC})$ for $\sigma_{ABC}$ in Eq.~\eqref{ea:essentialTRI}. It follows from Lemma~\ref{le:unique} that for any $p\in(0,1)$, the separable states
$p\s_{AB} + (1-p) \proj{11}$, $p\s_{AC} + (1-p) \proj{11}$, and
$p\s_{BC}+(1-p)\proj{11}$ are uniquely compatible with the state
$p\s_{ABC}+(1-p)\proj{111}$. So we have generated a family of
separable bipartite marginals which are uniquely compatible with a
genuinely entangled state, extending Example~\ref{ex:main}.

We now generalize Example~\ref{ex:main} to a different family of states that satisfy the
conditions in Theorem~\ref{thm:main}. Let $\s_{ABC}$ be as in Eq.~\eqref{ea:essentialTRI}, and the product state $\ket{a,b} \in \cR(\s_{BC})\cap
\cR(\s_{BC}^\G)$. Such product state always exists because
$r(\s_{BC}) = r(\s_{BC}^\G) = 3$, and there is a product state in
any 2-dimensional two-qubit subspace. For example, we can choose
$\ket{a,b} = ({1\over\sqrt2}, {1\over\sqrt2}) \ox (\sqrt{1\over 3},
\sqrt{2 \over 3})$. We have the following corollary now.

\bcr \label{cr:mixture} Let $\vec{p}=(p_1,\cdots,p_n)$,
$\sum^n_{i=1} p_i =1$, $p_1>0$, $p_i\ge0$. For $i>1$ suppose the
product states $\ket{a_i,b_i}\in\cR(\s_{BC})\cap \cR(\s_{BC}^\G)$
where $\ket{a_i}$ is real and $\s_{BC}$ is the reduced density
operator of the state $\s_{ABC}$ in Eq.~\eqref{ea:essentialTRI}. The three
reduced density operators of the three-qubit state $\s_{\vec{p}} =
p_1 \s_{ABC} + \sum^n_{i=2} p_i \proj{a_i,a_i,b_i}$ are only
compatible with genuinely entangled states. \ecr

\bpf It is sufficient to show that the three reduced density
operators $(\s_{\vec{p}})_{AB}$, $(\s_{\vec{p}})_{AC}$, and
$(\s_{\vec{p}})_{BC}$ satisfy the three conditions (i), (ii), and
(iii) in Theorem~\ref{thm:main}. Recall that $\s_{ABC}$ satisfies these
conditions. Since $\ket{a_i,a_i} \in \cR(\s_{AB})$ and
$\ket{a_i,b_i} \in \cR(\s_{BC})=\cR(\s_{AC})$, we have $\cR ((
\s_{\vec{p}})_{XY}) = \cR(\s_{XY})$ for any $X,Y\in\{A,B,C\}$. So
condition (i) is satisfied. Next, the same argument shows that
$(\s_{\vec{p}})_{AB}$ has birank $(3,4)$, which is exactly condition
(iii). Third, the hypothesis $\ket{a_i,b_i}\in\cR(\s_{BC})\cap
\cR(\s_{BC}^\G)$ and $\ket{a_i}$ is real imply that the birank of
$(\s_{\vec{p}})_{BC}$ is $(3,3)$. So condition (ii) is also
satisfied. \epf

\end{document}